\documentclass[a4paper,11pt]{article}

\usepackage{fouriernc} 

\usepackage[english]{babel}   
\usepackage[T1]{fontenc} 

\usepackage[margin=25mm]{geometry}

\usepackage{calc}
\usepackage{xcolor}
\usepackage{graphicx}
\usepackage{verbatim}

\usepackage[colorlinks=true,
                linkcolor=teal,
                citecolor=olive,
                destlabel=true,
                bookmarks=true]{hyperref}

\usepackage{xurl} 
\hypersetup{breaklinks=true}

\usepackage{bookmark} 

\usepackage[style=numeric, url=false]{biblatex}
\usepackage{csquotes} 

\DeclareBibliographyDriver{eprint}{%
  \usebibmacro{bibindex}%
  \usebibmacro{begentry}%
  \usebibmacro{author/editor+others/translator+others}%
  \setunit{\printdelim{nametitledelim}}\newblock
  \usebibmacro{title}%
  \newunit
  \printlist{language}%
  \newunit\newblock
  \usebibmacro{byauthor}%
  \newunit\newblock
  \usebibmacro{byeditor+others}%
  \newunit\newblock
  \printfield{howpublished}%
  \newunit\newblock
  \printfield{type}%
  \newunit
  \printfield{version}%
  \newunit
  \printfield{note}%
  \newunit\newblock
  \usebibmacro{organization+location+date}%
  \newunit\newblock
  \usebibmacro{eprint}%
  \newunit\newblock
  \usebibmacro{addendum+pubstate}%
  \setunit{\bibpagerefpunct}\newblock
  \usebibmacro{pageref}%
  \newunit\newblock
  \iftoggle{bbx:related}
    {\usebibmacro{related:init}%
     \usebibmacro{related}}
    {}%
  \usebibmacro{finentry}}
\usepackage{amsmath}
\usepackage{amsfonts}
\usepackage{amssymb}

\let\N\naturals

\let\R\reals

\newcommand{\defeq}{\mathrel{\vcenter{\hbox{\scriptsize:}}\!=}} 

\usepackage{amsthm}

\swapnumbers        

\usepackage{etoolbox}

\makeatletter
  \patchcmd{\@endtheorem}{\@endpefalse }{}{}{} 
  \patchcmd{\endproof}{\@endpefalse}{}{}{}
\makeatother

\newtheorem{paber@thm}{}[section]

  \theoremstyle{plain}
  
    \newtheorem{theorem}[paber@thm]{Theorem}
    \newtheorem*{theorem*}{Theorem}

    \newtheorem*{thm*}{Theorem}

    \newtheorem*{prop*}{Proposition}

    \newtheorem{proposition}[paber@thm]{Proposition}
    \newtheorem*{proposition*}{Proposition}

    \newtheorem{corollary}[paber@thm]{Corollary}
  
    \newtheorem{lemma}[paber@thm]{Lemma}
    \newtheorem*{lemma*}{Lemma}

  \theoremstyle{definition}

    \newtheorem{definition}[paber@thm]{Definition}
    \newtheorem*{definition*}{Definition}

    \newtheorem*{defn*}{Definition}

  \theoremstyle{remark}

    \newtheorem*{rmk*}{Remark}
  
    \newtheorem{remark}[paber@thm]{Remark}
    \newtheorem*{remark*}{Remark}

    \newtheorem{example}[paber@thm]{Example}
    \newtheorem*{example*}{Example}

\newcommand\oracle{\mathcal{O}}
\newcommand\market{\mathcal{M}}
\newcommand\action{\mathcal{S}}
\newcommand\labels{\mathcal{K}}

\newcommand\Expectation{\mathbb{E}}
\newcommand\Probability{\mathbb{P}}
\newcommand\Utility{\mathbf{U}}
\newcommand\beliefs{\mathcal{R}}
\newcommand\sell{\mathtt{sell}}
\newcommand\arb{\mathtt{arb}}

\usepackage{circuitikz}
\title{Do backrun auctions protect traders?}
\author{Andrew W. Macpherson \thanks{This work was funded by the Flashbots Research grants programme as FRP-34. GitHub: \url{https://github.com/flashbots/mev-research/}}}

\begin{document}
\maketitle

\begin{abstract}

  We study a new `laminated' queueing model for orders on batched trading venues such as decentralised exchanges.
  The model aims to capture and generalise transaction queueing infrastructure that has arisen to organise MEV activity on public blockchains such as Ethereum, providing convenient channels for sophisticated agents to extract value by `acting on' end-user order flow by performing arbitrage and related HFT activities.
  In our model, market orders are interspersed with orders created by arbitrageurs that under idealised conditions reset the marginal price to a global equilibrium between each trade, improving predictability of execution for liquidity traders.
  
  If an arbitrageur has a chance to land multiple opportunities in a row, he may attempt to manipulate the execution price of the intervening market order by a probabilistic `blind sandwiching' strategy.
  To study how bad this manipulation can get, we introduce and bound a \emph{price manipulation coefficient} that measures the deviation from global equilibrium of local pricing quoted by a rational arbitrageur.
  We exhibit cases in which this coefficient is well approximated by a `zeta value' with interpretable and empirically measurable parameters.

\end{abstract}

\vspace{3ex}

\begin{quote}

  \emph{Market orders and stop orders face a common risk: those who submit market orders, or whose stop orders convert to market orders, anticipate that there will be robust and orderly quoting and trading activity to provide an immediate execution at a reasonable price.}
  
  \hfill SEC Memorandum, \cite{sec2016certain}.

\end{quote}

\setcounter{tocdepth}{1}
\tableofcontents

\section{Introduction}\label{introduction}

Arbitrage is the means by which markets arrive at competitive equilibrium \cite{nau1991arbitrage}.
In the case that markets are fragmented across many trading venues, arbitrageurs communicate local prices between different venues and facilitate the convergence of the market on a global equilibrium, alleviating local information asymmetries.
A robust and orderly functioning of arbitrage is a basic assumption for much of financial economics \cite{modigliani1958cost,werner1987arbitrage,varian1987arbitrage}.

Arbitrage opportunities in liquid markets are short-lived, so speed and fine-grained control over execution ordering is of the essence for arbitrage traders.
Naturally, in the information age arbitrage is primarly the domain of high-frequency algorithmic traders \cite[Part II]{yadav2015algorithmic}.

The act of competing for arbitrage opportunities at high frequency has been argued to have some negative externalities \cite{yadav2015algorithmic,budish2015high,daian2020flash}, notably:
\begin{itemize}
  \item 
    A winner-takes-all latency race incentivizes massive investment in physical infrastructure and low latency algorithm design and implementation.
    Because this infrastructure and R\&D labour is private and highly specialised, much of it is wastefully duplicated.
    
  \item 
    The enormous barrier to entry for new agents erected by this investment favours a monopolistic market structure.

  \item
    The extraordinary efficiency with which HFT algorithms can react to signals can amplify the effect on markets of modelling errors or sudden changes in liquidity structure, e.g.~flash crashes \cite[1628]{yadav2015algorithmic}.

  \end{itemize}

What if the competition for arbitrage opportunities could itself be made more orderly?
The advent of programmable, transparent trading environments on public blockchains such as Ethereum provides a setting for radical experimentation in this direction.
Moreover, the extreme proliferation of novel assets and trading venues on these domains makes the need for efficient arbitrage even more immediate.

\subsection{Bringing order to the arbitrage market}

%
The discrete time environment of programmable blockchains provides a new type of opportunity for arbitrageurs: the possibility to \emph{backrun} a price-moving trade by having the arbitrage transaction sequenced in the very next position in the block.
This differs from the `continuous' execution environment of traditional exchanges in that no other transaction can intervene between the target trade and the backrun.
%
%
This type of backrun is commonly employed to arbitrage CFMM DEXes \cite{wang2022cyclic}.

The general model for how arbitrageurs carry out this procedure is as follows: unconfirmed transactions from `ordinary' end users gather in a pool --- public or otherwise --- observed by algorithmic traders, or \emph{MEV searchers}.
Information about the contents of these transactions may be completely or selectively revealed to searchers, who then construct arbitrage transactions and attempt to arrange for them to be sequenced so as to best exploit the opportunities created by the incoming order flow.

In Ethereum's public mempool, the contents of pending transactions are fully visible to arbitrageurs, who must compete to be allocated each backrun position through complex and subtle bidding and timing games \cite{daian2020flash}.
Moreover, if arbitrageurs themselves submit transactions to the public mempool, they too can become the victim of targeted exploitation by other searchers.
Like traditional HFT, this incurs substantial costs and barriers to entry on the part of searchers; furthermore, the congestion caused by excessive messaging and failed transactions is an additional cost borne by Ethereum itself.

Given the vital function performed by arbitrageurs in aligning prices across different trading venues and the problems associated with the ad hoc public mempool competition, it is not surprising that the industry has begun to explore designs for preferred channels along which arbitrageurs can compete for and act on backrun opportunities in a more orderly fashion.\footnote{\emph{MEV Blocker}, \url{https://mevblocker.io}}\footnote{\emph{MEV-Share}, \url{https://collective.flashbots.net/t/mev-share-programmably-private-orderflow-to-share-mev-with-users/1264}}\footnote{SUClave, \url{https://github.com/getclave/suclave-ethglobal-istanbul}}
Such channels allow the competition for positioning to happen out of band, so that Ethereum need process \emph{only one} searcher transaction (or string of transactions originating from the same entity) after each user transaction, in recognition of the fact that generally only one arbitrage transaction is needed to reset the price after each price-moving trade.

\paragraph{Laminated batches}
We now attempt to formalise the alternation of user and searcher arbitrage orders that arise from the type of privileged searcher channel just described.
Consider an idealised trading venue $\market$ on which two classes of trader place orders to exchange a risky asset with a num\'eraire:
\begin{itemize}
  \item \emph{liquidity traders}, or \emph{price takers}, who are opinionated about the size (denominated in the risky asset) they wish to trade, but although they surely wish to get the best possible price given their size and time requirements, are not opinionated about the exact value of that price;
  \item \emph{arbitrageurs}, who are opinionated about the marginal price at which they are prepared to trade, but not size.
\end{itemize}
Liquidity and arbitrage orders are collected in two separate queues and alternately executed, as follows:
\begin{enumerate}
  \item 
    First, a set of liquidity orders $\sell(r_i)$ of (signed) sizes $\{r_1,\ldots,r_K\}$ are enqueued.
    When executed, the resulting trades will have some price impact that depends on $r_i$ and the market liquidity curve immediately before execution.
  
  \item
    A sequence of $K$ or $K+1$ execution slots are then made available for arbitrageurs to place orders: one immediately after each liquidity order and preceding the next (if any) and, optionally, one at the front of the batch.
    Depending on how we structure the discipline, arbitrageurs may have varying degrees of information about their slot allocation and the number, contents, and ordering of liquidity orders when they make their decisions.

    An arbitrage order $\arb(p)$ is specified by declaring the target price $p$, indicating that its originator commits to buy/sell all liquidity available at a price better than $p$.\footnote{In the presence of transaction fees, arbitrageurs would quote separate bid and ask prices, introducing a spread.}
    Executing an arbitrage order with target price $p$ resets the marginal price on to $\market$ $p$.

\end{enumerate}
Orders are then executed in sequence
\[
 [\quad \arb(p_0),\quad \sell(r_1),\quad \arb(p_1),\quad \ldots,\quad \sell(r_K),\quad \arb(p_k) \quad]
\]
alternating between liquidity and arbitrage orders.
We refer to this process of interleaving orders as \emph{lamination}, and to the resulting ordered segment as a \emph{laminated batch}.
For an illustration of how this model can be specialised to model a two-lane queueing system actually deployed in the wild, see Example \ref{mev-blocker}.

\begin{remark}[Implementations of laminated batches]

  The semantics of arbitrageur orders in our model is equivalent to that of a large partial fill limit order.
  For trading on any CFMM DEX, this order format can be programmed into a wrapper contract that computes the correct sizing in terms of the reserve balances at execution time.
  
  In practice, an arbtrageur/taker lamination queue can be implemented by providing two endpoints to submit orders, one for liquidity traders and one for arbitrageurs, together with a mechanism to assign arbitrage slots --- say, an auction, or a stake-weighted random election.
  These endpoints could either be integrated into the venue itself, say, as a monolithic smart contract ensemble, or as an additional infrastructure layer over an existing `vanilla' CFMM that enqueues liquidity and arbitrage orders in two lanes and delivers them in an atomic, contiguous bundle.
  For our analysis to apply in the latter case, orders that arrive at the CFMM other than via the lamination layer must be considered as beyond the strategy horizon of our arbitrageurs.
  
\end{remark}

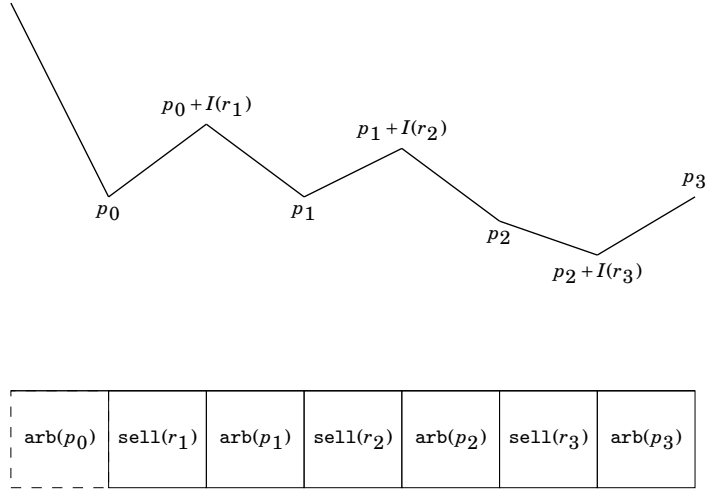
\begin{figure}[!ht]
  \centering
  \vspace{3ex}
  \resizebox{0.6\textwidth}{!}{%
  \begin{circuitikz}[label distance=-1ex]
  \tikzstyle{every node}=[font=\tiny]
  \draw [, line width=0.2pt , dashed] (-4,-1) rectangle node {$\arb(p_0)$} (-3,0);
  \draw [, line width=0.2pt ]         (-3,-1) rectangle node {$\sell(r_1)$} (-2,0);
  \draw [, line width=0.2pt ]         (-2,-1) rectangle node {$\arb(p_1)$} (-1,0);
  \draw [, line width=0.2pt ]         (-1,-1) rectangle node {$\sell(r_2)$} (0, 0);
  \draw [, line width=0.2pt ]         (0, -1) rectangle node {$\arb(p_2)$} (1, 0);
  \draw [, line width=0.2pt ]         (1, -1) rectangle node {$\sell(r_3)$} (2, 0);
  \draw [, line width=0.2pt ]         (2, -1) rectangle node {$\arb(p_3)$} (3, 0);

  \node (2) at   (-4, 0) {};
  \node (3) at   (3, 0) {};
  \node (4) at   (-4, 4) {};
  \node (5) at   (-3, 2) {};
  \node  (6) at  (-2, 2.75) {};
  \node  (7) at  (-1, 2) {};
  \node [label=above:$p_1+I(r_2)$]  (8) at  (0, 2.5) {};
  \node  (9) at  (1, 1.75) {};
  \node [label=below:$p_2+I(r_3)$] (10) at (2, 1.4) {};
  \node  (11) at (3, 2) {};
  \node  (12) at (-4, 4) {};
  \node [label=below:$p_0$] (13) at (-3, 2) {};
  \node  (14) at (-2, 2.75) {};
  \node [label=below:$p_1$] (15) at (-1, 2) {};
  \node [label=below:$p_2$] (17) at (1, 1.75) {};
  \node [label=above:$p_3$] (18) at (3, 2) {};
  \node  (19) at (-2, 2.75) {};
  \node  (20) at (-2, 2.75) {};
  \node [label=above:$p_0+I(r_1)$] (21) at (-2, 2.75) {};
  \draw (2.center) to (3.center);
  \draw (4.center) to (5.center);
  \draw (5.center) to (6.center);
  \draw (6.center) to (7.center);
  \draw (7.center) to (8.center);
  \draw (8.center) to (9.center);
  \draw (9.center) to (10.center);
  \draw (10.center) to (11.center);
  \end{circuitikz}
  }%
  
  \vspace{2ex}
  \label{fig:batch}
  \caption{Price evolution in a laminated batch. $I(r_i)\defeq$\;price impact of trade $\sell(r_i)$.}
\end{figure}

To analyse our batch trading model, we make a few simplifying assumptions:
\begin{itemize}
  \item All markets are frictionless, that is, without transaction fees.
  \item Arbitrageurs can trade instantly on a reference market $\oracle$ with zero price impact.
  \item The target venue $\market$ as a whole is a price taker from $\oracle$. That is, $\market$ is sufficiently small compared to $\oracle$ that order flow on $\market$ does not influence the reference price $p_\oracle\leftarrow\oracle$.
  \item Batches are committed instantaneously from the perspective of the reference market, so the global equilbrium price $p_\oracle$ during the execution of a single batch is constant.
\end{itemize}
Given these assumptions, how would a rational arbitrageur set prices?

\paragraph{Passthrough arbitrage}
Suppose that each arbitrageur is \emph{myopic}, so that they treat each arbitrage opportunity in complete isolation --- that is, they have a strategy time horizon of just one trade.
This models an idealised situation in which arbitrage is so competitive that each agent treats as negligible their chances of being able to land multiple spots.
Then an arbitrageur's best strategy is always to set their target price to the reference price $p_\oracle$, absorbing all liquidity available on $\market$ at a better price, and simultaneously trading an equal amount in the opposite direction on $\oracle$.\footnote{In practice, arbitrageurs may also have more abstract methods to acquire opposite exposure to their trade on $\market$.}
We call this strategy \emph{passthrough pricing}, because it simply passes through information from the reference market to $\market$ without introducing a bias.

With passthrough arbitrage, each liquidity order on the verge of execution finds the market in a state of global price equilibrium.
There are a number of reasons to find this dynamic desirable:
\begin{itemize}    
  \item \emph{Ordering independent.}
    It reduces the problem of predicting executions to that of modelling the block price, which is indexed over a single discrete time variable.
  \item \emph{Sandwich-resistant.}
    Liquidity traders can place market orders with greater confidence that they will not receive an `manipulated' execution price.
  \item \emph{Price oracle.}
    Prices set by arbitrageurs provide an incentive-compatible internal price oracle.
  \item \emph{Mitigating wasteful competition.}
    Constraining the strategy space (say, by eliminating sandwiches) and simplifying execution of arbitrage MEV lowers the barrier to entry for would-be arbitrageurs and reduces the scope for model failures.

\end{itemize}
Passthrough pricing also has the effect of isolating adverse selection effects on liquidity providers, known as LVR \cite{milionis2022automated}, to the top of the block, which may facilitate redistribution mechanisms aimed at mitigating these effects.

Unfortunately, these results depend on the clearly unrealistic assumption that arbitrageurs make decisions about each trade in isolation.
What happens when this assumption is lifted?

\paragraph{Manipulation}
If an arbitrageur $i$ knows that he has a non-negligible chance to be allocated two or more slots --- say, slot $k-1$ and $k$ --- in a row, then under suitable conditions he may choose to use this opportunity to try to `sandwich' the intervening liquidity order $r_k$ so that the latter executes at a manipulated price, increasing the overall profit.
This strategy is harder to execute in the absence of execution guarantees about atomicity and ordering and without information about the contents of the target liquidity order.
However, with suitable models for the distribution of these unknowns, the strategy can still be dominant in expectation.
That is, trading in laminated batches is \emph{not} fully sandwich resistant.

The main goal of this paper is to study how bad this manipulation can get in the presence of an incompletely decentralized pool of rational arbitrageurs.
We can measure this by introducing a \emph{manipulation coefficient} $C_\market\defeq \log(\hat p/p_\oracle)\in\mathbb{R}$, where $\hat p$ is the price set by an arbitrageur at Nash equilibrium and $p_\oracle$ is the reference price.
Obtaining bounds or concentration results on $C_\market$ should be useful for exchange designers, who might wish to convince traders that manipulation on their venue will remain within certain limits, given some ansatz about the liquidity structure, the distribution of liquidity orders $r_i$, and that of the allocation function of opportunities to arbitrageurs.

\begin{remark}[Sybil resistance]

  It may be tempting to imagine that we can simply enforce `myopic' arbitrage by a hardcoded trading rule that the same arbitrageur cannot be allocated two consecutive slots.
  However, it is generally very difficult to rule out collaboration or Sybil identities that allow `many' arbitrageurs to act as one, so such rules would carry little water in practice.

\end{remark}

\subsection{Results}

Our main contributions are to formalise the model sketched in the previous section, establish limiting cases --- in particular, when the probability of landing a sandwich vanishes --- in which arbitrageurs provide passthrough pricing, and establish a closed form approximation to price manipulation near the limit.

Suppose given a market $\market$, a random order flow $(r_k)_{k=1}^K$, a population $\{1,\ldots,N\}$ of rational arbitrageurs, and a random slot allocation $\alpha:[K]_+\defeq\{0,\ldots,K\}\rightarrow[N]$.
Each arbitrageur $i$ quotes a target price $p_i$ with associated price $\phi(s_i)$ that will be the target price in slot $k$ for all $k\in\alpha^{-1}(i)$.\footnote{The case where arbitrageurs may set different prices for different slots will be treated shortly.}

Although our model permits the use of an arbitrary random allocation $\alpha$, our \emph{results} turn out to only depend on $\alpha$ through its single and binary marginal distributions.
That is, if we define
\begin{itemize}
  \item the \emph{primary weight} $a_{i,k} \defeq \Probability[\alpha(k)=i]$, the probability that player $i$ is assigned slot $k$;
  \item a \emph{secondary weight} or \emph{coupling} $b_{i,k} \defeq \Probability[\alpha(k-1)=i \mid \alpha(k)=i]$, the probability that, given player $i$ is assigned slot $k$, they are also assigned the slot immediately before (hence giving an opportunity to sandwich).
\end{itemize}

Our first result is that under general conditions, best response strategies are dominant, so the optimal target price for each arbitrageur can be computed separately.

\begin{theorem}[Existence and convergence of equilibrium prices, \ref{theorem-existence}]

  The marginal quote $\hat{p}$ at each slot depends only on $p_\oracle$, the marginal distributions of the trade size and direction for each $k$, and the primary and secondary allocation weights $a_{i,k},b_{i,k}$.
  
  Moreover, $\hat{p}$ converges to $p_\oracle$ as either
  \begin{enumerate}
    \item the trade distributions $r_k\stackrel{p}{\rightarrow}0$ in probability for all $k$.
    \item the secondary allocation weights $b_{i,k}\rightarrow 0$ for all $i,k$.
  \end{enumerate}

\end{theorem}

We also obtain an explicit equation \eqref{lamination-equation} for the dominant strategy in terms of the allocation weights and an expected price impact function.

As a special case, we obtain a formalisation of the claim made in the previous section that myopic arbitrageurs provide passthrough pricing.

\begin{corollary}[Myopic arbitrageurs provide passthrough pricing, \ref{theorem-myopic}] 

  Suppose that all secondary allocation weights vanish.
  Then the dominant strategy for each arbitrageur is to set prices to $p_\oracle$.

\end{corollary}

The general form \eqref{lamination-equation} for the \emph{lamination equation} obtained for the optimiser $\hat{p}$ is quite complicated.
It can be approached under simplifying assumptions. Suppose:
\begin{enumerate}
  \item 
    The order flow distribution $r$ is \emph{symmetric} in the index set $\{1,\ldots,K\}$. 
    Intuitively, this means arbitrageurs have no information about the absolute ordering of liquidity orders within the batch.
  \item 
    Slots are allocated to arbitrageurs by independent weighted lotteries (that are also independent of $\vec{r}$).
    That is, each participating arbitrageur $i$ has some weight $0\leq w_i\leq 1$, where $\sum_iw_i=1$, and slots are allocated to $i$ by Bernoulli trials (weighted coin tosses) with weight $w_i$.
\end{enumerate}
Under these conditions, approximating the log price impact function by its linearisation yields an approximation to the equilibrium price manipulation coefficient in terms of a `zeta function'
\[
  Z_{\phi,r}(w)\defeq \frac{1-w}{1-M_r(\lambda)w}
\]
where $M_r$ is the moment generating function of the distribution of $r$ and $\lambda=\frac{d}{dx}\log\phi(1) > 1$ is a weighting determined in terms of the liquidity curve of the market.\footnote{In a Balancer-style weighted 2-asset CPMM pool, the quantity $\lambda$ is the reciprocal of the pool share of the num\'eraire.}

\begin{remark}
  In practical situations we can expect the error to be quite small ($<1\%$), more than enough to make order-of-magnitude judgements.
\end{remark}

\begin{theorem}[Zeta function approximation to manipulation coefficient, \ref{theorem-zeta}]

  Suppose that the liquidity curve $\phi$ of $\market$ is approximated by its linearisation at $p_\oracle$ with error at most $C$.
  Then
  \[
    \log(\hat p/p_\oracle) \approx Z_{\phi,r}(\check w)
  \]
  with error at worst $C$.

\end{theorem}

\paragraph{Per-slot pricing}
The reader familiar with existing backrun services such as MEV-Blocker and MEV-Share may find it surprising that we require our backrunners to each choose a single target price per batch, rather than allowing them the freedom to craft separate backrun transactions for each slot.
There are a few reasons for this:
\begin{itemize}
  \item 
    Giving arbitrageurs the ability to set different target prices for different slots in the same batch requires us to posit an additional labelling of the slots that players can use to declare their moves, complicating the game design.
  \item
    More choices for arbitrageurs generally means more leeway to perform manipulation strategies.
    Conversely, the most powerful bounds on manipulation are likely to be available with the most constrained action space for MEV actors.
\end{itemize}

Nonetheless, we study this more complicated model briefly in \S\ref{multi-price}, formulating a generalised laminated queue model in which liquidity orders and arbitrage slots are labelled by an auxiliary set of strings $\labels$ and arbitrageurs may quote a different price for each label.
Labels are mapped to execution positions by a random indexing $\mathrm{idx}:\labels\rightarrow[K]$
We find that under suitable blinding assumptions for backrunners, the Nash equilibrium computation for this labelled lamination game collapses to the single-price one.

\begin{theorem}[\ref{theorem-multi}]

  For sufficiently small order flow or coupling, and suitable action space, each arbitrageur has a dominant strategy.
  If the mapping from labels to execution positions is uniformly random, this dominant strategy is to set all prices to the same value: the optimum price of the associated uniform marginal price game.

\end{theorem}

\subsection{Related work}

A plethora of approaches have been floated to combat the combined menace of wasteful MEV market structures and exposing liquidity traders to unpredictable, disequilibrium, or manipulated outcomes.
It is beyond our scope here to give a full overview of extant models; we list here only a few that are particularly close in spirit to the laminated batch model.

\paragraph{Exotic queue disciplines}
The lamination model presented in this paper and the probabilistic methods used to analyse it are inspired by the very general queue theoretic approach to blockchain based markets introduced in \cite{macpherson2023adversarial}.

In \cite{ferreira2022credible}, a `greedy sequencing rule' is introduced with a similar goal of ensuring that prices remain `close' to a constant value throughout the batch, with all directional trades isolated to one end of the block.
%
%
In comparison to that work, the dynamics of price movements in the laminate model is a natural product of incentive structures rather than hardcoded into the queue discipline.

\paragraph{Private mempools}
Perhaps the most basic approach to combatting price manipulation by frontrunning is to impose pre-confirmation privacy.
While it may seem intuitively clear that private transactions benefit from frontrunning protection, it is a matter of debate whether this always leads to improved expectation outcomes for their originators \cite{schoneborn2009liquidation,marshall2023false}.

Privacy, encoded as randomness in the distribution of market orders from the arbitrageur's perspective, is part of our model.
The arguments of this paper explore some of its effects and limitations.
As we confirm in \S\ref{multi-price}, sufficient privacy does indeed hobble some of the most powerful targeted value extraction strategies.
However, even with very limited information about unconfirmed liquidity trades --- just an ansatz about their distribution and the market structure --- profitable price manipulation is still possible, negatively impacting trader outcomes.

\paragraph{Frequent batch auctions}
A simple and popular way to completely eliminate ordering based MEV is to ensure that all trades commute by aggregating them together into batches and executing them at the same \emph{uniform clearing price} (UCP).
The batch auction approach has been the subject of substantial recent interest \cite{budish2015high, cowprotocoloverview, gong2023order, canidio2023arbitrageurs}.
In \cite{canidio2023arbitrageurs}, the authors derive strong conclusions about their batched CF-AMM similar in spirit to the present paper: in our language, that in a perfectly competitive and frictionless market, arbitrageurs deliver passthrough pricing to the batch.
However, unlike in the lamination model, those arbitrageurs obtain zero revenue in equilibrium.

The most popular implementation of blockchain-settled batch auctions, the CoW protocol \cite{cowprotocoloverview}, also attempts to distinguish between makers and takers with the presence of `liquidity orders,' which at least in spirit echoes the laminate model's bipartite queues.

Another difference between lamination batches and traditional UCP batches is that traders do \emph{not} get uniform \emph{execution} prices, only uniform pre-execution \emph{marginal} prices.
Larger trades still have a greater instantaneous market impact and therefore generally receive a worse execution price.
Conversely, in a UCP batch auction, the price impact of larger trades is effectively socialised across smaller trades.

\subsection{Future directions}
For practical applications and relevance to real-world instantiations of laminated batches, it would be useful to lift some of our hypotheses:
\begin{itemize}
  \item
    If we want our model to eventually apply to exchanges that account for a significant fraction of volume for a particular asset pair, we will have to lift the assumption that $\market$ is a price taker. 
  \item
    To apply the model to limit order books or CFMMs with a multi-block horizon, we will need to allow a dynamic liquidity structure.    
  \item
    Our model assumes that all market orders are denominated in the same asset (the `risky' one).
    More realistic dynamics could be obtained in a symmetric model that does not prefer a particular asset as the num\'eraire and allows orders to be denominated in any asset.
\end{itemize}
It would also be interesting to analyse the effect on incentives of auction-based mechanisms for realising $\alpha$, and reward redistribution systems such as volume-weighted or position-dependent fees and rebates.

\paragraph{Acknowledgements}
The author is grateful to Quintus Kilbourn and Evan Kim for valuable feedback on early drafts of this work.

\section{Model}
\subsection{Markets} \label{markets}

Our model of a trading venue follows the approach of \cite{milionis2023complexity}.

\begin{definition} \label{market}
  A \emph{market} \(\market\) consists of the data of a \emph{price density function} \(\phi_\mathcal{M}:U\rightarrow\mathbb{R}\) defined on an open subset \(U\subseteq(0,\infty)\) which is piecewise \(C^\infty\) (possibly with discontinuities) and monotone decreasing. The parameter of the function is called the \emph{liquidity depth}. The market is said to be \emph{smooth} if \(\phi\) is everywhere \(C^\infty\) and \emph{invertible} if it is strictly monotone. 
\end{definition}

The data of a market can be interpreted as follows.
Suppose we have a risky asset \(A\) and a num\'eraire \(B\) that can be traded. 
Suppose that both assets are arbitrarily divisible and that a quantity \(r\) of \(A\) can be sold for \(\hat{C}(r)\) of \(B\), where \(\hat{C}\) is some piecewise-differentiable function defined in a neighbourhood \(U_0\) of \(r=0\). 
Then setting \(U=U_0+x\), where \(x>0\) is chosen arbitrarily so that \(U\subseteq(0,\infty)\), and \(\phi(u):=d\hat{C}(r)/dr(u-x)\) defines a price density curve in the sense of Definition \ref{market}. 

\begin{example}[CFMM]

  Let \(f:(0,\infty)^2\rightarrow\R\) be a CFMM (where $f$ is \(C^\infty\) with generically positive partial derivatives) and \(\lambda\in\R\) a level. 
  The indifference set \(f^{-1}(\lambda)\) is an embedded submanifold of \((0,\infty)^2\) that projects diffeomorphically onto an open subset of either axis. 
  Then there exists an open set \(U\subseteq\R\) such that the projection \(f^{-1}(\lambda)\rightarrow (0,\infty)\) admits a section \(s:U\rightarrow f^{-1}(\lambda)\). 
  Composing this section with the other projection gives a smooth real-valued function \(P:U\rightarrow (0,\infty)\). 
  Then \(\phi_f:=-dP/dt\) defines a market on reserve set \(U\) in the sense of Definition \ref{market}.

\end{example}

\begin{example}[Reference market] 
  
  The highly liquid reference market \(\oracle\) arises as a limiting case where \(\phi:(0,\infty)\rightarrow\mathbb{R}\) is constant and the liquidity depth $x_0 \gg 0$ is large compared to all other quantities under consideration.

\end{example}

\begin{definition}[Action] \label{cost}
  
  Fix an initial market depth \(x\in U\). Since $\phi$ is monotone, it is integrable over intervals within its domain, so we can define an \emph{(absolute) cost function} or \emph{action}
  \begin{align*}
    \action(x,y)&\defeq \int_x^y \phi(u)du \\
    &= \Phi(y) - \Phi(x)
  \end{align*}
  where $\Phi$ is an antiderivative of $\phi$.
  It is the revenue of a trader that moves the liquidity depth on $\market$ from $x$ to $y$ (by selling $y-x$ units of \(A\)). 
  When \(y<x\), \(\action(x,y)<0\) and $\action(y,x)=-\action(x,y)$ is the cost to \emph{buy} $x-y$ units of \(A\).

\end{definition}

If $\phi$ is defined by a CFMM $f$ and level $\lambda$, then $\action(x_0,x)$ is a constant plus the amount of num\'eraire needed to balance an amount $x$ of the risky asset.

A cyclic arbitrage that moves the reserve depth from $x$ to $y$ consists of buying $x-y$ units of $A$ on $\market$ and selling the same number of units on $\oracle$.
It results in a gain of \(\action(x,y) + p_\oracle\cdot(x-y) \) in the num\'eraire.
We will need to refer to this quantity a lot, so we introduce notation for it.

\begin{definition}[Opportunity cost]

  The \emph{relative} or \emph{opportunity cost} function at \(x_0\in U\) is the quantity 
  \begin{align}
    C(x) &\defeq \phi(x_0)\cdot (x-x_0) + \action(x,x_0). 
  \end{align}
  It is the cost of an arbitrage that moves the reserve depth from $x_0$ to $x$.
  Both terms have the same sign, and if \(x>x_0\) (resp.~if \(x<x_0\)), the linear term (resp. nonlinear term) dominates.
  Hence this function is continuous, convex, and valued in non-negative reals with a minimum at \(C(x_0)= 0\). 
  If $\phi$ is $\mathcal{C}^k$, then $C$ is $\mathcal{C}^{k+1}$ with derivatives 
  \begin{align}
    \label{cost-derivative}
    C'(x) &= \phi(x_0)-\phi(x)\\ 
    \label{cost-second-derivative}
    C''(x) &= -\phi'(x) \qquad (>0).
  \end{align}
  In words, the derivative of opportunity cost is marginal price differential.
  Since it is convex, it is also locally integrable on the domain of $\phi$.

\end{definition}

Note that the opportunity costs at different liquidity depths \(x_0,x_1\) differ by a linear term with gradient \(\phi(x_1)-\phi(x_0)\).

\begin{remark} 
  
  In economic terms, the quantity \(C\) can be interpreted as a surplus supply or demand of the risky asset \(A\). 
  If \(x_1>x_\mathcal{O}\), then it is a surplus supply (of the risky asset), i.e.~\(\mathcal{M}\) will sell this amount below the odds.
  Correspondingly, if \(x<x_\mathcal{O}\), it is a surplus demand.

\end{remark}

\subsection{Game}\label{game}

We consider a game of \(N\) players \(X_1,\ldots,X_n\) parametrised by the following data:
\begin{enumerate}
  \item A market $\market$ with liquidity structure $\phi:U\rightarrow (0,\infty)$.
  \item A natural number $K\in\N$. We write \([K]_+=\{0\}\sqcup [K]\).
  \item A sequence \(r_1,\ldots,r_K\) of real numbers, representing market orders. 
  \item A map $\alpha:[K]_+\rightarrow[N]$.
  \item A positive real number $x_0>0$.
  \item A positive real number $x_\oracle>0$. We write $p_\oracle\defeq \phi(x_\oracle)$.
\end{enumerate}
In situations of imperfect information, we will consider $x_0$, the $r_i$, and $\alpha$ as random variables.
For simplicity, $K$ and $x_\oracle$ will generally be fixed (i.e.~known to all players).
Each player's beliefs about the other parameters comprise a distribution on
\[
  \mathcal{H}_{N,K} = \R\times \R^K \times N^{K_+}.
\]
with the three factors representing the space of $x_0$, $r$, and $\alpha$, respectively.

We recall also the definitions of the \emph{primary} and \emph{secondary} allocation weights
\[
  a_{i,k}\defeq \Probability[\alpha(k)=i] \qquad b_{i,k}\defeq\Probability[\alpha(k-1)=i\mid \alpha(k)=i]
\]
where $b_{i,k}$ is defined whenever $a_{i,k}>0$.

\begin{remark}[Generality of the information assumptions]
  
  No effective generality is lost by the assumption that \(N\) and \(K\) are known to all players, since any situation with an unknown (but bounded) number of players and trades can be represented by a suitable random allocation function $\alpha$ from a fixed large $K$ into a fixed large $N$.

\end{remark}

\paragraph{Actions and utilities} Each player has action space $A\subset (0,\infty)$.
To simplify the treatment of indices, we define $s_0\defeq x_0$, $r_0\defeq 0$, and write $\alpha(k)=0$ for $k=-1$.
That is, the `zeroth market order' has zero size, and the zeroth player `nature' plays $x_0$ --- the starting liquidity --- in the slot immediately preceding the first `real' slot.
Given an action profile $\vec{s}\in A^N$, the utility of player $i$ is
\begin{align} \label{utility}
  \Utility_i\left(\vec{s}\right) &\defeq \sum_{\alpha(j)=i} C(s_{\alpha(j-1)} + r_j) - C(s_i) \\
  &= \sum_{\alpha(j)=i}\left[\action(s_{\alpha(j-1)} + r_j, s_i) + p_\oracle\cdot (s_{\alpha(j-1)} + r_j - s_i)\right].
\end{align}
So that expectations are defined, we will assume that $|C(s+r)|$ is bounded by an integrable function on all of $A$ and the range of $r$.

\subsection{Interpretation} \label{interpretation}

The interpretation of this game is as follows: $N$ players insert arbitrage trades \(s_{\alpha(0)},\ldots,s_{\alpha(K)}\) interleaving \(K\) market orders $(r_1,\ldots,r_K)$ on a market \(\market\) between assets \(A\) and \(B\) (the numéraire). 
The numbers $r_i$ represent the size and direction of the orders, denominated in \(A\), with positive $r_i$ indicating a sell order.
The liquidity curve $\phi_\market$ does not change during this sequence of trades; this hypothesis is plausible for a CFMM DEX but not a limit order book.
If $\alpha(n)=i$, then player $n$ will have a chance to insert an arbitrage in the $n$th slot, where slot $0$ is top of block and slot $k$ is immediately after the $k$th trade for $k>0$.

A player's beliefs about the distribution of $r_i$ are a forecast of the size and direction of incoming trades.
His beliefs about $\alpha$ could derive from common knowledge that the allocation is uniformly pseudo-randomly derived from cryptographically secure entropy source, or it could be a forecast about the outcome of some mechanism (e.g.~an auction).

The absolute boundedness assumption on $C(s+r)$ means that the marginal price remains bounded independently of the actions of arbitrageurs and traders.
Clearly, this assumption is verified on any market that can be implemented on a finite computer.
It is violated on some idealised continuum models, a key example being a CPMM with $A=(0,\infty)$.
To deal with this case, we should instead bound $A$ away from $0$, i.e.~$A=(\epsilon,\infty)$ for some strictly positive $\epsilon$.

In play, each player \(X_i\) chooses a \emph{target depth} \(s_i\). 
If the market is invertible, this is equivalent to choosing a marginal market \emph{price} \(\phi(s_i)\), i.e.~creating an arbitrage $\arb(\phi(s_i))$.
Otherwise, the target depth is the more fundamental parameter since this determines directly the amount that must be traded.

All players may also trade on a reference market \(\oracle\) at a constant price \(p_\oracle\).
In \S\ref{markets} we have established that the payout for moving the liquidity depth on $\market$ from $x$ to $y$, then trading the same amount in the opposite direction on $\oracle$ is $C(x)-C(y)$.

For simplicity we assume that utility for players is measured only in terms of asset \(B\). 
That is, players do not gain utility from holding \(A\) in inventory.
To gain utility, players must complete a pure profit cyclic arbitrage as above.

\subsection{Information structures} 

Our model is general enough to capture a broad range of important hypotheses about the information about order flow and slot allocation available to arbitrageurs when they make their commitments.

\begin{definition}[Monopoly] \label{monopoly}

  A backrun game is a \emph{monopoly} if there exists an $n\in[N]$, the \emph{monopolist}, such that the image of $\alpha$ is $\{n\}$ almost surely.

\end{definition}

\begin{definition}[Freeness] \label{locally-free} \label{free}

  A backrun game is \emph{free} if $\alpha$ is injective with unit probability.
  That is, for all $i,j\in [K]_+$ we have $i\neq j \Rightarrow \Probability(\alpha(i)=\alpha(j))=0$.
  In particular, for all $i,j\in [K]_+$ and $n\in[N]$, $i\neq j \Rightarrow \Probability(\alpha(j)=n\mid \alpha(i)=n)=0$.\footnote{The adjective `free' is chosen by analogy with the notion of a free, or non-interacting theory in physics. Intuitively, distinct backrun opportunities do not interact with one another.}
  
  It is \emph{locally free} if $\alpha$ (a.s.) does not map any two consecutive slot indices to the same player.
  It is equivalent to say that the couplings $b_{i,k}=\Probability(\alpha(k-1)=i|\alpha(k)=i)$ vanish for all positive $a_{i,k}=\Probability(\alpha(k)=i)$.

\end{definition}

\begin{definition}[Symmetry]

  The order vector \((r_i)_{i=1}^K\) is  \emph{symmetric} if the $r_i$ are identically distributed, that is, their marginal distributions are all the same. 
  It is \emph{strictly symmetric} if they are symmetrically distributed, that is, their joint distribution is invariant under the permutation action of $\Sigma_K$.
  A strictly symmetric order vector is symmetric.

  The allocation $\alpha$ is \emph{symmetric} if its distribution is invariant under the right action of $\Sigma_{K+1}$ on $[N]^{[K]_+}$.

\end{definition}

Intuitively, a [strictly] symmetric order vector, resp.~allocation, means that the agent whose perspective we share --- in this paper, the arbitrageur deciding on a price target --- has no information about the ordering of liquidity orders, resp.~the ordering of the allocation.
They still may have information about the content of orders, for example, the unordered set of orders from which the batch will be drawn, or the allocation weights, for example the number of slots they will be allocated.
They may even have information about the contents of the orders they will backrun.

\begin{definition}[Blind allocation] \label{blind}

  A laminated queue has \emph{blind allocation} if $\alpha(k)$ and $r_k$ are independent random variables for all $k=1,\ldots,K$.
  The condition extends automatically to the case where $k$ is itself a $[K]$-valued random variable.
\end{definition}

\begin{remark}

  Intuitively, a laminated queue can be non-blind if the mechanism that determines $\alpha$ `knows' about $r$.
  For example, this can happen if $\alpha$ is determined by an auction among arbitrageurs, and arbitrageurs have some information about liquidity orders when they make their bids.
  So for example, the MEV-Blocker allocation is non-blind, while MEV-Share allocation can in principle be blind if the maximum privacy settings are used.

  Our definition also permits that liquidity traders know about $\alpha$ before creating their orders.
  This can be the case for single-player bundling mechanisms where the arbitrageur also chooses $\alpha$.

  Note that a monopolistic allocation is always blind (because a deterministic random variable is always independent of any other random variable).

\end{remark}

\begin{example}[Randomly permuted order list] \label{permuted-order-list}

  Let $r=(r_i)$ be a vector of order distributions, and let $\sigma\sim\mathrm{Unif}(\Sigma_K)$ be a uniformly random permutation independent of $(r_i)$.
  Then the random vector $\sigma^*r$ with $i$th coordinate $r_{\sigma (i)}$ is strictly symmetric.

\end{example}

\begin{example}[Randomly permuted allocation] \label{permuted-allocation}

  Similarly, if \(\alpha\) is any random allocation and \(\sigma\sim\mathrm{Unif}(\Sigma_{K+1})\) is independent of \(\alpha\), then \(\alpha\circ\sigma\) is symmetric.

\end{example}

\begin{example}[Known matching of players to orders, unknown ordering] \label{known-matching-unknown-ordering}

  Suppose $(r_i)$ is defined, as in Example \ref{permuted-order-list}, as a permutation of a known vector $(r'_i)\in\R^K$ by a uniformly random permutation $\sigma$.
  Let $\alpha':[K]_+\rightarrow [N]$ be a known allocation, and define $\alpha\defeq\alpha'\circ\sigma$.
  Then $\alpha$ is symmetric, and player $n$ is allocated the slots backrunning the orders $\{r_j\mid \alpha(j)=n\}$ with probability $1$.
  See also \S\ref{multi-price}.

\end{example}

\begin{example}[Independent allocations] \label{independent-allocation}

  Let $\alpha'_k$, $k=0,\ldots,K$ be a sequence of i.i.d.~$[N]$-valued random variables with pmf $w$.
  That is, each player $i=1,\ldots,N$ flips a $w(i)$-weighted coin to land a move in slot $k$ for each $k=0,\ldots,K$.
  Then $\alpha'$ is symmetric.
  If $\sigma$ is a random permutation (not necessarily uniform), then $\alpha'\circ\sigma$ is also a vector of i.i.d.~random variables.

\end{example}

\subsection{Solution} \label{solution}

The general formula for expected payoffs splits as a weighted sum of conditional expectations:
\begin{align} \label{expected-utility}
  \Expectation (\Utility_i(\vec{s})) &= %
    \sum_{k=0}^K \Probability(\alpha(k)=i) \cdot \big( %
      \underbrace{\Expectation[C(s_{\alpha(k-1)}+r_k)\mid \alpha(k)=i]}_\text{backrun revenue} - %
      \underbrace{\Expectation[C(s_i)\mid \alpha(k)=i]}_\text{frontrun cost} %
    \big) \\
    &= \sum_{k=0}^K a_{i,k}\Expectation\Utility_{i,k}(\vec{s}),
\end{align}
where the $k$th term 
\begin{align}
  \Expectation\Utility_{i,k}(\vec{s}) \defeq %
    & \Expectation\left[C(s_{\alpha(k-1)}+r_k)\mid \alpha(k)=i\right] - C(s_i)
\end{align}
is the expected payoff of being allocated slot $k$ given $\vec{s}$ (when $a_{i,k}>0$).

The expected backrun revenue in slot $k$ depends on player $i$'s strategy $s_i$ only when player $i$ also wins the previous slot, i.e.~$\alpha(k-1)=i$.
The contribution to the optimisation problem of the backrun revenue summand of $\Expectation\Utility_{i,k}(\vec{s})$ is therefore weighted by the secondary allocation weights $b_{i,k}$.
Note that $b_{i,0}=0$ for all $i$, since by definition $\alpha(-1)=0$ is not a real player.

We therefore compute
\begin{align}
  \Expectation\Utility_{i,k}(\vec{s}) &= b_{i,k}\cdot \Expectation[C(s_i+r_k)\mid \alpha(k)=i] - C(s_i) + \{\text{terms not depending on }s_i\}.
\end{align}
The summand $\Expectation[C(s_i+r_k)\mid \alpha(k)=i]$ is the backrun or end slice revenue from running a successful sandwich.
Note that the conditioning drops out of the expression if we assume that the allocation is blind in the sense of Definition \ref{blind}.

For the purposes of finding a best response strategy to $\vec{s}_{\hat{i}}$, we can disregard the last term and simply try to optimise the sum of expressions 
\begin{equation}
  \widetilde{\Expectation\Utility}_{i,k}(s) \defeq b_{i,k}\cdot \Expectation[C(s+r_k)] - C(s).
\end{equation}
Since $|C(s+r_k)|$ is by hypothesis bounded independently of $s$ and $r$, we may move the derivative inside the expectation and obtain 
\begin{align}
  \Expectation\Utility'_i(s\mid\vec{s}_{\hat{i}}) &= \sum_{k=0}^K a_{i,k} \cdot\widetilde{\Expectation\Utility}'_{i,k}(s) \\
  &= \sum_{k=0}^K a_{i,k} \cdot\left( b_{i,k}\cdot \Expectation[C'(s+r_k)\mid \alpha(k)=i] - C'(s) \right) \\
  \label{expected-utility-final} 
  &= \sum_{k=0}^K a_{i,k} \cdot\left(\phi(s) - b_{i,k}\cdot \Expectation[\phi(s+r_k)\mid \alpha(k)=i] + (b_{i,k}-1)\phi(x_\oracle)\right). 
\end{align}
Note that this formula depends only on the marginal distributions of the $r_k$ and on the \emph{pairwise} joint distributions of the $\alpha(k)$, not on the full joint distribution.
Therefore, optimising by setting $\Expectation\Utility'_i(s\mid\vec{s}_{\hat{i}})=0$ gives us the \emph{lamination equation }
\begin{equation} \label{lamination-equation}
  \left(\sum_{k=0}^K a_{i,k}\right) \cdot (\bar{\phi}(s^*) - 1) = \sum_{k=1}^K a_{i,k}b_{i,k}\left(\Expectation[\bar{\phi}(s^*+r_k)\mid \alpha(k)=i] - 1  \right), \qquad i=1,\ldots,K
\end{equation}
where we have written $\bar{\phi}(s)\defeq \phi(s)/ p_\oracle$ for the \emph{price manipulation factor}.
Note that the index $0$ does not appear on the right hand side, because $b_{i,0}=0$.

Clearly, the solution $s^*$ to this equation, if it exists, does not depend on the strategies $s_{-i}$ of other players, and satisfies $\bar\phi(s^*) \rightarrow 1$ under either of the following r\'egimes:
\begin{enumerate}
  \item $r\stackrel{p}{\rightarrow} 0$.
  \item $b_{i,k}\rightarrow 0$.
\end{enumerate}
By a continuity argument, a solution ${s^*}$ exists in an $L^0$-neighbourhood of $\{(b_{i,k})_{k=1}^K=0\}\cup\{r \stackrel{\mathrm{a.s.}}{=} 0\}$.

\begin{theorem} \label{dominant-strategy-unique} \label{theorem-existence}
  
  Assume $\phi$ is continuously differentiable.
  For sufficiently small secondary allocation weights or sufficiently small order flow, there exists an action space $A\subset(0,\infty)$ containing $x_\oracle$ such that the unique solution in $A$ to $\Expectation\Utility'_i(s)=0$ is a dominant strategy.
  \qedhere

\end{theorem}
\begin{proof}

  We have established that $\Expectation\Utility_i(s_i\mid \vec{s}_{\hat i})$ has a local maximum ${s^*}$ close to $x_\oracle$ for small $(b_{i,k})_{k=0}^K$ or $(r_k)_{k=0}^K$.
  By continuity, it is strictly convex on some neighbourhood $A$ of ${s^*}$ containing $x_\oracle$.

\end{proof}

\begin{corollary} \label{dominant-strategy-passthrough} \label{theorem-myopic}

  Suppose that $\alpha$ is locally free (Definition \ref{locally-free}), that is, player $i$ almost surely lands no two consecutive slots.
  Then player $i$'s unique dominant strategy is passthrough pricing.

\end{corollary}
\begin{proof}

  Immediate from the form \eqref{expected-utility-final} of expected utility. \qedhere

\end{proof}

\subsection{Symmetric trade distribution with independent slot allocation}

When the information structure is highly symmetric and uncoupled, the formulas for optimal arbitrage simplify considerably.
Under such hypotheses we can derive approximations to quotes in a simple closed form that can be used to derive bounds on manipulation.
For the remainder of this section, we suppose:
\begin{itemize}
  \item Liquidity order flow $(r_k)_{k=1}^K$ is symmetric (arbitrageurs have no information about the sequence order);
  \item Allocations are blind (arbitrageurs do not know more about the trades they are backrunning than they do about a generic trade);
  \item Slot allocations for player $i$ are determined by Bernoulli trials with weight $w_i=a_{i,k}=b_{i,k}$. 
  That is, allocations are independent in the sense of Example \ref{independent-allocation}, hence symmetric.
\end{itemize}
Then the lamination equation \eqref{lamination-equation} simplifies to
\begin{equation}
  \bar\phi({s^*}) - 1 = \frac{K}{K+1} w_i(\Expectation[\bar\phi({s^*}+r)]- 1). \label{optimization-symmetric}
\end{equation}
In the interests of economy of notation, we absorb the `top of block' factor into the allocation weight, writing $\check{w}_i\defeq \frac{K}{K+1} w_i$ henceforth.
The sandwich term $\bar\phi(s^*+r)$ makes it hard to solve this equation for ${s^*}$ directly.
Instead, we proceed under an exponential approximation that turns out to have very small error in practice.
We also assume in this section that sizes are normalised so that equilibrium liquidity depth is unity $x_\oracle=1$.

\begin{definition}[Exponential price impact]

  Suppose $\bar\phi(x_\oracle + x)=\bar\phi(1+x)=\exp(-\lambda x)$ for some coefficient $\lambda>0$.
  Then setting RHS of equation \eqref{optimization-symmetric} splits the product 
  \[
    \Expectation(\bar\phi(s+r)) = \bar\phi(s)\cdot \Expectation(\bar\phi(x_\oracle + r))
  \]
  of the manipulation factor by the expected price impact factor at equilibrium of the order flow $r$.

  We can also express the expected price impact at equilibrium
  \[
    \Expectation(\bar\phi(x_\oracle + r)) = \Expectation[e^{-\lambda r}] = M_r(\lambda)
  \]
  in terms of the moment generating function $M_r$ of the distribution of $r$.

\end{definition}

\begin{example}[CPMM]

  If $\market$ is a weighted CPMM with invariant $x^\alpha y^\beta$, we have $\phi(x)=x^{-(1+\alpha/\beta)}$.
  The exponential approximation at $x=1$ is determined by 
  \[
    \lambda = \frac{d}{dx}\log\phi(1) = -(1+\alpha/\beta),
  \]
  that is the reciprocal of the pool share of the num\'eraire \cite{martinelli2019non}.

  The linearisation error is quite good for reasonable ranges of $r$: for example, in the unweighted case ($\lambda=2$) we can get a bound $C=0.011$ for $r\in[-0.1,0.1]$ or $C=0.000101$ for $r\in[-0.01,0.01]$.
  
\end{example}

\begin{remark}[Better linear approximations]

  With few changes, we could replace the linearisation at $x_\oracle$ with an optimal linear or affine $L^\infty$-approximation to $\log\bar\phi$ on its domain.
  However, we then lose the nice interpretation of the slope $\lambda$.

\end{remark}

If price impact is assumed exponential, we can gather like terms in equation \eqref{optimization-symmetric} to obtain
\[
  \bar\phi({s^*})\cdot\left(1-\check w_i \cdot M_r(\lambda)\right) = 1- \check w_i.
\]
The coefficient of $\bar\phi({s^*})$ is positive as long as the expected price impact factor $M_r(\lambda)$ is at most $\check w_i^{-1}$.
In this case we can divide through to find
\begin{equation} \label{exp-manipulation-factor}
  \bar\phi({s^*}) = Z_{\phi,r}(\check w_i) \defeq  \frac{ 1-\check w_i }{ 1-\check w_i\cdot M_r(\lambda) }.
\end{equation}
In general, when we have only $||\log\bar\phi(1+r) + \lambda r || < C$ for some $C>0$, this gives us the \emph{zeta function approximation} for $\bar\phi({s^*})$.

\begin{remark}
  
  If $w_i$ is large --- say, if player $i$ is a monopolist (Definition \ref{monopoly}) --- then the bound $M_r(\lambda)\check w_i < 1$ might be exceeded for realistic price impact factor expectations.

\end{remark}

\begin{theorem}[Zeta function approximation for price manipulation factor] \label{theorem-zeta}

  Suppose that linear approximation error of $\log\bar\phi(s)$ at $x_\oracle$ is bounded by $C$, that is, 
  \[
    |\log\bar\phi(x_\oracle + r) + \lambda r | < C
  \]
  on the domain of $\phi$.
  Then the zeta function approximation error of the manipulation factor $\log\bar\phi({s^*})$ is also at most $C$, that is,
  \[
    | \log\bar\phi({s^*}) - \log Z_{\phi,r}(\check w_i)| < C
  \]
  for $0\leq \check w_i\leq \min\{1,M_r(\lambda)^{-1}\}$.
  

\end{theorem}

\begin{remark}[Sizing $r$]
  Note that the relevant figure here is the expected price impact of each individual liquidity trade.
  If traders tend to split their larger orders up into many small orders, bounds can be easily satisfied without constraining the trading volume.
\end{remark}

\paragraph{Numerical computations} Expanding the zeta function approximation via the Newton-Mercator series, we find
\begin{align}
  \log\bar\phi({s^*}) &\approx \log(1-\check w_i) - \log ( 1-\check w_i\cdot M_r(\lambda) ) \\
  &= \sum_{n=1}^\infty \frac{1}{n}\left(\check w_i^n - (\check w_i\cdot M_r(\lambda) )^n \right) \\
  &= \sum_{n=1}^\infty \frac{1}{n}\check w_i^n\cdot \left(1  - M_r(\lambda)]^n \right)
\end{align}
We can obtain bounds on these power series by bounding $\check w_i$ and $\left|M_r(\lambda)-1\right|$ away from $1$.
In practice it is reasonable to bound these quite small so as to get an estimate of $\log\bar\phi({s^*})$ to a few decimal places.

\begin{example}[Bounding the log manipulation coefficient]

  Suppose $D>\max\{w_i,w_i\cdot M_r(\lambda)\}$.
  Then $\log\bar\phi(s) = w_i\cdot(1 - M_r(\lambda)) + O(D^2)$.
By itself, this bound is perhaps not as powerful as we would like; for example, if $D=0.1$ then we get an approximation error of up to $0.1$.
It becomes more so if you assume also that $Q=|M_r(\lambda)-1|$ is small, say $0.01$; we find, for example, that
\[
  \log\bar\phi({s^*}) = -2w_i^2 + O(D^3, QD, Q^2)
\]
with an error of less than $0.01$.

\end{example}

\subsection{Trader experience}
Suppose a liquidity trader makes an order $q$ which will be sequenced in slot $k\stackrel{\$}{\in}\{1,\ldots,K\}$.
Suppose backrunner allocations are independent as in Example \ref{independent-allocation}, where the $i$th player has weight $w_i\in[0,1]$ and internal order flow model $\sim r_i$.
Let $\bar p=\bar\phi(s_{\alpha(k-1)})$ be the normalised marginal price at the time that $q$ executes.

Our liquidity trader does not know player $i$'s model for the distribution of order flow, and so player $i$'s dominant strategy $s_i^*$ must be treated as a random variable.
Let us model beliefs of arbitrageurs about order flow distribution as a distribution over some parameter space $\beliefs$ of probability measures on $\R$.
Then $s^*$ and hence $\bar p$ can be regarded as a (deterministic) function on $[N]\times\beliefs$.
Equivalently, taking weights, we could model it as a function on $[0,1]\times\beliefs$.
We can use these parametrisations to make statements about the distribution of $\log\bar p$.

\begin{example}[Equal weights]

  Suppose that allocation weights are evenly distributed, i.e.~$w_i=1/N$ for all $i=1,\ldots,N$.
  Then the distribution of $(w,r)$ is supported on $\{1/N\}\times\beliefs$, and we can treat $\bar p$ as a deterministic function of the $\beliefs$-valued random variable $r$.
  
  It makes sense to ask how the value of $N$ contributes to $\bar p$ in this case.
  Under the zeta function approximation, we have $\bar p\approx Z_{\phi,r}\left(\frac{K+1}{KN}\right)$ as distributions with error bounded a.s.~independently of $N$.
  This random variable tends to $1$ a.s.~as $N\rightarrow\infty$.
  That is, unsurprisingly, as the arbitrageur population becomes more decentralised and the market power of any individual arbitrageur tends to zero, the marginal prices experienced by traders tend to the oracle price $p_\oracle$ with absolute error depending only on the error of the linear approximation to $\phi$.

\end{example}

\section{Generalisations}

Our core laminated batch model with market orders invites generalisations in a number directions that may make it more realistic.

\subsection{Per-slot pricing} 
\label{multi-price}

In real-world systems like MEV-Share and MEV-Blocker, algorithmic traders can propose different backrun order (bundles) for each target user transaction separately.
If elected to backrun a particular target transaction, the backrunner is guaranteed atomic execution of the target-backrun pair.
Backrunners can specify a target transaction in their order by a label (i.e.~hash) without necessarily having any information about the eventual execution ordering among target-backrun bundles. 

We can extend our trading model --- with liquidity orders standing in for user transactions and arbitrageurs for backrunners --- to capture this action space by introducing a random \emph{labelling} bijection $\mathrm{idx}:\labels\rightarrow [K]$ of the slots (excluding top of block).
The allocation function is replaced by a mapping $\alpha:\labels\rightarrow[N]$, and the liquidity trade vector is also indexed by $\labels$.
Each player $i$ chooses a $\labels$-indexed family of target prices $(s_{i,k})_{k\in\labels}$.
Then $X_i$'s payoffs are as if he played $s_{i,h}$ in the $\mathrm{idx}(h)$th slot for each $h\in\alpha^*(i)$.

Indexing by label $h\in\labels$, the payoffs are defined as follows:
\begin{equation}
  \Utility_i\left( (s_{i,h})_{h\in\labels} \mid \vec{s}_{\hat{i}} \right) \defeq %
    \sum_{h\in\alpha^*(i)} C(s_{\alpha(h'),h'} + r_h ) - C(s_{i,h})
\end{equation}
where $\mathrm{idx}(h') = \mathrm{idx}(h)-1$.

If we restrict the action space to the diagonal, we can reindex the labelling out of the utility formula and recover a game of the form considered in \S\ref{game}; call this the \emph{diagonal subgame}.

The utility optimisation equation takes a very similar form to \eqref{expected-utility-final}:
\begin{equation}
  \Expectation\Utility'_i\left( (s_{i,h})_{h\in\labels}\right) = %
    \sum_{h\in\labels}a_{i,h}\left(\sum_{h\in\labels}b_{i,h} \Expectation[C(s_{i,h'} + r_h )\mid \alpha(h)=i ] - C(s_{i,h}) \right).
\end{equation}
However, the randomness of the unknown indexing $\mathrm{idx}$ feeds into $h'$ and hence $s_{i,h'}$, so unlike the diagonal case this value must itself be treated as a random variable.

\begin{proposition} \label{theorem-multi}

  For sufficiently small allocation weights or order flow, there exists an action space $A\subset(0,\infty)^K$ containing $(x_\oracle,\ldots,x_\oracle)$ such that all players in this game have a unique dominant strategy.

  If the sequence order $\mathrm{idx}:\labels\rightarrow[K]$ is uniformly random, then the dominant strategy is the same as the dominant strategy for the diagonal subgame, that is, to set all $s_{i,h}$ to the same value.

\end{proposition}
\begin{proof}

  The first claim follows from along the same lines as the proof of Proposition \ref{dominant-strategy-unique}.
  For the second, by $\Sigma_K$-symmetry , if $\Expectation\Utility_i$ has a maximum on the diagonal for some allocation weights and order flow, then for any small deformation of the parameters it has a nearby maximum that remains on the diagonal.
  \qedhere

\end{proof}

In other words, despite its greater input complexity, if backrunners have no information about sequence order then per-slot pricing is no more expressive than the uniform price game discussed in \S\ref{game}.

\begin{example}[MEV-Blocker] \label{mev-blocker}

  Let's compare this model to the way MEV-blocker \cite{mevblocker2023docs} works.
  In MEV-Blocker, the $\mathtt{eth\_sendBundle}$ RPC allows a searcher to construct a transaction bundle with a given (unsigned partial) transaction in the first slot.
  As suggested by the documentation in \emph{loc.~cit.}, a typical use case is to use the rest of the bundle to backrun the target transaction.
  
  The full contents of transactions are visible to the searcher at the time they choose their strategy, so that $(r_h)_{h\in\labels}$ is deterministic.
  However, searchers are not informed of the final sequence order (unless they are integrated with the block builder), so $\mathrm{idx}:\labels\rightarrow[K]$ should be considered a random variable.
  
  The allocation function $\alpha:\labels\rightarrow[N]$ of slots to searchers is determined by an auction, so this random variable models the searcher's beliefs about the outcomes of this auction.
  A searcher's private value for an allocation $\alpha(h)=i$, and hence their bid in the allocation auction, depends on the amount they can expect to gain from backrunning the trade $r_h$ (or frontrunning the subsequent trade).
  This in turn depends on the value of $r_h$.
  That is, $\alpha(h)$ and $r_h$ are not independent, and the MEV-Blocker allocation function is not blind (in the sense of the obvious adaptation of Def.~\ref{blind} to the per-slot pricing game).
  
  This model assumes that all MEV-blocker bundles on the target market are executed consecutively, and not interspersed with any trades originating from other channels.

\end{example}

\begin{remark}[Inferring transaction ordering]

  Even if, as in MEV-Blocker, backrunners are not informed of transaction ordering when they make their decisions, they may be able to use knowledge of the contents of transaction contents to make forecasts about ordering. 
  For example, larger trades may be expected to land closer to the top of the block; concretely, for fixed $R\in\R$, $\Probability(r_k\leq R)$ is smaller for smaller $k$.
  The hypothesis that $\mathrm{idx}$ is uniformly random is unrealistic in this case.

\end{remark}

\subsection{Limit orders}

Allowing liquidity traders to place limit orders introduces $K$ additional parameters to the model.
The qualitative dynamics that result depend on whether the orders permit partial fills or not.
In the latter case, the payoff functions are discontinuous, so we get some interesting phase transitions in strategy space.
In particular, if the limit on a liquidity order is such that it would fail to execute at the oracle price, the optimal strategy for an arbitrageur may be to provide \emph{better} than equilibrium pricing so that the trade executes and they can profit from the backrun.

\begin{example}[Sandwiching a limit order]

  Suppose for simplicity that player $1$ is a monopolist, and $K=1$.
  This is the situation of an atomic sandwich of a single trade $(r,q)$.
  Suppose $r>0$ (so the liquidity trader wishes to sell the risky asset) and $q=\phi(y)$.

  The payoff function for sandwiching a partial fill limit is 
  \[
    C(r_0) - C(s) + C(\min\{q,s+r\}) - C(s).
  \]
  The mixed term is obtained by retracting the action space onto $A\cap(0,q-r]$.

  In the all or nothing case, we get a discontinuous payoff
  \[
    \Utility(s) \defeq C(r_0) - C(s) + \left\{ \begin{array}{cc}C(s+r) - C(s) & s+r\leq q \\ 0 & s+r> q \end{array} \right.
  \]
  and hence a nontrivial phase transition.
  
  Suppose $\sell(r,q)$ would fail to execute at the oracle price, i.e.~$x_\oracle + r = q + \epsilon > x_\oracle$ for some $\epsilon>0$. 
  For small $\epsilon>0$ the term $C(s+r)>0$ dominates the utility expression at $s= x_\oracle -\epsilon$, whence
  \[
    \lim_{s\rightarrow (x_\oracle-\epsilon)_+} \Utility(s) > \lim_{s\rightarrow (x_\oracle-\epsilon)_-} \Utility(s).
  \]
  That is, setting the price \emph{lower} than the oracle price has better payoff than passthrough pricing.
  Similar reasoning holds if we allow our sandwicher to set different prices for the first and second slot.

\end{example}

\subsection{Multiple batches} \label{multiple-batches}

Consider a repeated version of the backrun game where the same \(N\) players play \(M\in\mathbb{N}\sqcup\{\infty\}\) times, with \(x_0,\vec{r},\alpha,\phi,p_\oracle\) re-rolled each play.
Note that for realism in this case, we are forced to allow the liquidity curve to vary from turn to turn.
If we allow imperfect information about the liquidity curve, payoff analysis is substantially complicated.

The repeated game greatly opens up the strategy space, enabling new ways for arbitrageurs to coordinate to give liquidity traders bad prices.
For example, by analogy with the solution to the repeated prisoner's dilemma it is natural to search for Nash equilibria (in non-dominating strategies) in `price ring' strategies, where all arbitrageurs set prices to a level that deviates from the oracle price according to some scheme not depending on the player, unless another arbitrageur deviates first.

Furthermore, the possibility of using smart contracts to add new commitments to the multi stage game opens yet more possibilities for price rings to arise as `Stackelberg' equilibria \cite{hall2021game}.

\printbibliography

\end{document}